\documentclass[a4paper,toc=bibliography]{scrartcl}
\usepackage{lmodern}
\usepackage[T1]{fontenc}
\usepackage[latin1]{inputenc}
\usepackage[english]{babel}
\usepackage{amsmath,amsthm,braket,amssymb,tikz,caption,subcaption,color,booktabs,units,bbold,microtype}
\usetikzlibrary{decorations.pathmorphing}
\usetikzlibrary{decorations.markings}
\usetikzlibrary{knots}
\usetikzlibrary{patterns}
\usepackage{csquotes}
\usepackage{hyperref}
\usepackage{todonotes}
\usepackage{txfonts} 
\theoremstyle{plain}
\newtheorem{thm}{Theorem}
\newtheorem{cor}{Corollary}

\newtheorem{lem}{Lemma}

\theoremstyle{remark}

\newcommand{\de}{\, \mathrm{d}}

\newcommand{\asterisknum}{\addtocounter{equation}{1} \tag{\theequation}}
\DeclareMathOperator{\Tr}{Tr}

\DeclareMathOperator{\e}{e}
\DeclareMathOperator{\Li}{Li}
\numberwithin{equation}{section}
\newcommand{\Rr}{\mathbb{R}}
\newcommand{\Nn}{\mathbb{N}}
\newcommand{\Zz}{\mathbb{Z}}
\newcommand{\Cc}{\mathbb{C}}

\title{The free energy of the two-dimensional dilute Bose gas. II. Upper bound}

\author{Simon Mayer\thanks{\texttt{simon.mayer@ist.ac.at}} \ , \quad Robert Seiringer\thanks{\texttt{robert.seiringer@ist.ac.at}}\\[1ex]
	Institute of Science and Technology Austria (IST Austria)\\
	Am Campus 1, 3400 Klosterneuburg, Austria}

\date{May 19, 2020}

\begin{document}

\maketitle

\begin{abstract}
	We prove an upper bound on the free energy  of a two-dimensional homogeneous Bose gas in the thermodynamic limit. We show that for $a^2 \rho \ll 1$ and  $\beta \rho \gtrsim 1$ the free energy per unit volume differs from the one of the non-interacting system by at most $4 \pi \rho^2 |\ln a^2 \rho|^{-1} (2 - [1 - \beta_{\mathrm{c}}/\beta]_+^2)$ to leading order, where $a$ is the scattering length of the two-body interaction potential, $\rho$ is the density, $\beta$ the inverse temperature and $\beta_{\mathrm{c}}$ is the inverse  Berezinskii--Kosterlitz--Thouless critical temperature for superfluidity. In combination with the corresponding matching lower bound proved in \cite{DMS19} 
this shows equality in the asymptotic expansion.
\end{abstract}



\section{Introduction and main result}

\subsection{Introduction}

The first experimental observation of Bose--Einstein condensation in dilute alkali gases \cite{AEMWC95,DMAvDDKK95}, with the subsequent advances and activities in experimental and theoretical physics, has led to  renewed interest in the mathematical aspects of interacting Bose and Fermi gases. For an overview of some of the rigorous results on Bose gases obtained in recent years,  see \cite{FS19,BECbook,rougerie}.
The present article is a sequel to \cite{DMS19}, to which we refer for an extended introduction on the topic of the dilute Bose gas and further recent results and references.

We shall investigate  the free energy  of a dilute Bose gas in the thermodynamic limit at positive temperature. Recall that in three spatial dimensions, the free energy (per unit volume) as a function of the inverse temperature $\beta = 1/T$ and the particle density $\rho$ satisfies the asymptotic identity
\begin{equation} \label{eqn-3d-free-energy-asymptotics-ub}
f^{\text{3D}}(\beta,\rho)  = f_0^\text{3D}(\beta,\rho) + 4 \pi a \rho^2 \left( 2 - \left[ 1 - \left( \frac{\beta_\text{c}^\text{3D}(\rho)}{\beta} \right)^{3/2} \right]_+^2 \right) (1 + o(1))
\end{equation}
as $a^3 \rho \to 0$, where $f_0^\text{3D}(\beta,\rho)$ is the free energy density of an  ideal Bose gas, $a\geq 0$ is the scattering length of the interaction potential, $[\, \cdot \,]_+ = \max\{0,\cdot \}$ denotes the positive part and $\beta_\text{c}^\text{3D}(\rho) = \zeta(3/2)^{2/3} / (4\pi \rho^{2/3})$ is the inverse critical temperature for Bose--Einstein condensation (of the ideal Bose gas). The  proof of  \eqref{eqn-3d-free-energy-asymptotics-ub} was given in \cite{Seiringer2008} (lower bound) and \cite{Yin2010} (upper bound). The formula is valid in the regime $\beta \rho^{2/3} \gtrsim 1$ or, in other words, if the temperature $\beta^{-1}$ is of the order of the critical temperature of the ideal gas, or smaller. 

The main goal of this article is to complete the analysis for the (first two terms of the) free energy asymptotics of the Bose gas in \emph{two} spatial dimensions. We shall prove the upper bound
\begin{equation} \label{eqn-2d-free-energy-asymptotics-ub}
f^{\text{2D}}(\beta, \rho) \leq f_0^\text{2D}(\beta,\rho) + \frac{4 \pi \rho^2}{|\ln a^2 \rho|} \left( 2 - \left[ 1 - \frac{\beta_\text{c}^\text{2D}(\rho,a)}{\beta} \right]_+^2 \right) (1 + o(1))
\end{equation}
as $a^2 \rho \to 0$, where $\beta_\text{c}^\text{2D}(\rho,a)$ is the inverse Berezinskii--Kosterlitz--Thouless critical temperature for superfluidity  \cite{B71,B72,K74,KT73}, given by
\begin{equation} \label{eqn-def-beta-c}
\beta_\text{c}^\text{2D}(\rho,a) = \frac{\ln |\ln a^2 \rho|}{4 \pi \rho}.
\end{equation}
In combination with the corresponding lower bound proved in \cite[Theorem~1]{DMS19} we deduce that \eqref{eqn-2d-free-energy-asymptotics-ub} is actually an equality.

At first sight  \eqref{eqn-2d-free-energy-asymptotics-ub} and \eqref{eqn-3d-free-energy-asymptotics-ub} look  similar, but there are two important differences. The first one is the inverse of the logarithmic factor $|\ln a^2 \rho|$ appearing as a prefactor in the second term, which is particular to the two-dimensional system and already known from the asymptotics of the ground state energy  \cite{LY2001,Schick71}.  The second one concerns  the  inverse critical temperature $\beta_{\mathrm{c}}^\text{2D}(\rho,a)$, which in two dimensions depends  on the interaction  via its scattering length, and diverges in the dilute limit $a^2\rho \to 0$, which is not the case in three dimensions. Recall that the Mermin--Wagner--Hohenberg theorem \cite{Hohenberg67,MerminWagner66} forbids Bose--Einstein condensation at positive temperature in  two-dimensional systems, hence their behavior can be expected to be rather different from their three-dimensional analogues. 
These  differences are among the reasons why proving the free energy asymptotics in two dimensions is not merely a simple extension of the three-dimensional case. 

In the remainder of this section, we define the free energy in the thermodynamic limit, recall some facts about the ideal Bose gas,  and state our main result, Theorem~\ref{thm-ub}. Since in the following we will exclusively  deal with the two-dimensional system, we will omit the superscript ``2D'' in  the  free energies $f^\text{2D}$ and $f_0^\text{2D}$ and in the inverse critical temperature $\beta_{\mathrm{c}}^\text{2D}(\rho,a)$.

\subsection{The model} \label{subsec-model}

We consider the Hamiltonian for $N$ interacting bosons in a two-dimensional flat torus $\Lambda$, given by
\begin{equation}  \label{eqn-def-N-body-Hamiltonian}
H_N = \sum_{i=1}^N -\Delta_i + \sum_{i<j}^N v(d(x_i,x_j)),
\end{equation}
where $\Delta$ is the Laplacian on $\Lambda$, $d(x,y)$ is the distance function on the torus and $v \geq 0$ is a nonnegative two-body potential with finite scattering length $a>0$. We assume that $v$ is a measurable function that is allowed to take the value $+\infty$, which is appropriate to model hard disks. For a definition of the scattering length, we refer to \cite{DMS19,BECbook,LY2001} or to Sec.~\ref{ss:scatt} below. Having a finite scattering length is known to be equivalent to $v(|x|)(\ln|x|)^2$ being integrable outside a ball, see \cite{Seiringer12}.

The Hamiltonian $H_N$ acts on the Hilbert space $\mathcal{H}_N$, the symmetric tensor product of square integrable functions on the torus,
\begin{equation}\label{def:HN}
\mathcal{H}_N = \bigotimes_{\text{sym}}^N L^2(\Lambda).
\end{equation}
As a concrete realization of $\Lambda$ we shall  use  the square of side length $L$ embedded into the plane $\Rr^2$ with opposite sides identified, in which case $\Delta$ is the Laplacian on $\Lambda = [0,L]^2$ with periodic boundary conditions. The distance is then 
\begin{equation}
d(x,y) = \min_{k \in \Zz^2} | x - y - kL |.
\end{equation}

At inverse temperature $\beta = 1/T$ and average  particle density $\rho$, the free energy per unit volume is defined as
\begin{equation} \label{eqn-def-free-energy}
f(\beta,\rho) = - \frac 1 \beta \lim_{\substack{N,L \to \infty\\ N/L^2 = \rho}} \frac{1}{L^2} \ln \Tr_{\mathcal{H}_N} \e^{-\beta H_N},
\end{equation}
where the limit is the usual thermodynamic limit of large volume and large particle number with fixed density $\rho$. For a proof of the existence of the   limit in \eqref{eqn-def-free-energy} we refer to  \cite{Robinson71,Ruelle69}. We will give an upper bound on the free energy that is asymptotically exact in the dilute limit  where  $a^2 \rho$ is small while the dimensionless parameter $\beta \rho$ is of order one or larger. In physical terms, this means that the scattering length is small compared to the average particle distance,  while  the thermal wave length is of the same order as the average particle distance or larger.

For an ideal, i.e., non-interacting Bose gas in two dimensions,  the free energy is explicitly given by 
\begin{equation}
f_0(\beta,\rho) =  \frac{\rho}{\beta } \ln \left( 1 - \e^{- 4 \pi \beta \rho} \right) - \frac{1}{4 \pi \beta ^2} \Li_2 \left( 1 - \e^{-4\pi \beta \rho} \right) ,
\end{equation}
where
\begin{equation}
\Li_2(z) = - \int_0^z \frac{\ln(1-t)}{t} \de t
\end{equation}
is the polylogarithm of order two. It satisfies  the scaling relation $f_0(\beta,\rho) = \rho^2 f_0(\beta \rho,1)$. For later use, we  also recall the chemical potential
\begin{equation}
\mu_0(\beta,\rho) = \frac \partial{\partial\rho} f_0(\beta,\rho) = \frac{1}{\beta} \ln \left( 1 - \e^{- 4 \pi \beta \rho} \right).
\end{equation}

\subsection{Main theorem} \label{subsec-thm}

Our main result is an asymptotic upper bound on the free energy of the interacting system in terms of the free energy of ideal bosons and a correction term originating from the interaction, in case $a^2 \rho$ is small and $\beta \rho$ is fixed or large. This is the two-dimensional analogue of  \cite[Thm~1]{Yin2010}. We use the standard notation $x \lesssim y$ if there exists a constant $C>0$ such that $x \leq C y$ (and analogously for ``$\gtrsim$'').

Let $\rho_{\mathrm{s}}$ denote the superfluid density
\begin{equation} \label{eqn-rho_s-upper-bound}
\rho_\mathrm{s} = \rho \left[ 1 - \frac{\beta_{\mathrm{c}}(\rho,a)}{\beta} \right]_+
\end{equation}
with the inverse critical temperature $\beta_\mathrm{c}(\rho,a)$ defined in \eqref{eqn-def-beta-c}.

\begin{thm}[Upper bound on the free energy of the two-dimensional dilute Bose gas] \label{thm-ub}
	Assume that the interaction potential $v$ is nonnegative and has a finite scattering length $a$. In the limit $a^2 \rho \to 0$ with $\beta \rho \gtrsim 1$  fixed or large, we have
	\begin{equation} \label{eqn-thm-ub}
	f(\beta,\rho) \leq f_0(\beta,\rho-\rho_{\mathrm{s}}) + \frac{4 \pi }{|\ln a^2 \rho |} \left( 2 \rho^2  -  \rho_{\mathrm{s}}^2 \right) (1 + o(1))
	\end{equation}
	with
	\begin{equation}\label{o1}
	o(1) \lesssim \frac{\ \ln |\ln a^2 \rho |}{ |\ln a^2 \rho|}.
	\end{equation}
\end{thm}

We emphasize that \eqref{o1} is uniform in $\beta \rho$ for $\beta\rho \gtrsim 1$. Its order of magnitude agrees with the expected next order term at zero temperature \cite{And,Fou,Mor}, although with the wrong sign.  Moreover, the error term depends on the interaction potential $v$ only through its scattering length $a$, except for an additional error term of the form 
\begin{equation}
\frac 1 {|\ln a^2\rho|^2}  \int_{|x|\geq a (C a^2 \rho)^{-1/2} |\ln a^2\rho|^{-7/2}} v(|x|) [\ln(|x|/a)]^2 \de x
\end{equation}
for some $C>0$. This term is negligible compared to the main error term of the order $\ln|\ln a^2\rho|/ |\ln a^2\rho|^2$, but is non-uniform in $v$ and cannot be estimated solely in terms of $a$. 

The proof of Theorem~\ref{thm-ub} follows a very different route than the corresponding result in three dimensions in \cite{Yin2010}. It is in fact much shorter, and in many ways much simpler. Moreover, it works for a larger class of interaction potentials. Our proof strategy would not work in the three-dimensional case, however. This is due to the fact that in three dimensions the main correction term compared to $f_0$ is much smaller than in two dimensions; the relevant small parameter $a\rho^{1/3}$ enters linearly in \eqref{eqn-3d-free-energy-asymptotics-ub}, while in two dimensions $a\rho^{1/2}$ enters only logarithmically. Hence a greater accuracy is required in the analysis of the three-dimensional case, and several estimates employed here would be too crude to achieve this accuracy. 

One readily checks that
\begin{equation}
f_0(\beta,\rho-\rho_{\mathrm{s}})  - f_0(\beta,\rho) =  - \frac{1}{\beta} \int_0^{\rho_{\mathrm{s}}} \ln \left( 1 - \e^{- 4 \pi \beta (\rho - r)} \right) \de r \lesssim  \frac{\rho^2}{{|\ln a^2 \rho |}} \frac 1{\left( \ln |\ln a^2 \rho |\right)^2},
\end{equation}
which is much smaller than the main correction term of order $\rho^2 / |\ln a^2 \rho |$ in \eqref{eqn-thm-ub}, but is much larger than our bound on the $o(1)$ term. Hence our upper bound could also be stated with $f_0(\beta,\rho)$ in place of $f_0(\beta,\rho-\rho_{\mathrm{s}})$ on the right side, but at the expense of a larger error term. 

In combination with the lower bound proved in \cite{DMS19}, Theorem~\ref{thm-ub} establishes equality in the asymptotic expansion of the free energy:

\begin{cor}[Free energy asymptotics of the two-dimensional dilute Bose gas] 
	Assume that the interaction potential $v$ is nonnegative and has a finite scattering length $a$. In the limit $a^2 \rho \to 0$ with $\beta \rho \gtrsim 1$  fixed or large, we have
	\begin{equation} 
	f(\beta,\rho) = f_0(\beta,\rho) + \frac{4 \pi \rho^2}{|\ln a^2 \rho |} \left( 2 - \left[ 1 - \frac{\beta_\mathrm{c}(\rho,a)}{\beta} \right]_+^2 \right) (1 + o(1)),
	\end{equation}
	where
	\begin{equation}
	|o(1)| \lesssim \frac{\ln \ln |\ln a^2 \rho |}{\ln |\ln a^2 \rho|}.
	\end{equation}
	Here, $[ \, \cdot \, ]_+$ denotes the positive part and the inverse critical temperature $\beta_\mathrm{c}(\rho,a)$ is defined in \eqref{eqn-def-beta-c}.
\end{cor}

The bound on the $o(1)$ term stated in \eqref{o1} originates from the lower bound in  \cite[Theorem~1]{DMS19}, the one obtained here in the upper bound is smaller.  

The proof of Theorem~\ref{thm-ub} is given in the next section. It is split into several subsections for better readability, and starts with a brief outline of the strategy.

\section{Proof of the main theorem} \label{sec-proof}

\subsection{Outline of the proof strategy}

We use the Gibbs variational principle for the free energy and insert a suitable trial state into the free energy functional to obtain an upper bound. We partition the square $[0,L]^2$ into $(L/\ell)^2$ smaller boxes of size $\ell - R_0$, separated a suitable distance $R_0$, and construct the trial state as  a tensor product of identical (up to translation) trial states on each smaller box. While $\ell$ will be chosen large when $a^2\rho$ is small, it is independent of $L$. This has the advantage of having a smaller number of particles to deal with, allowing for simpler estimates. On the other hand, it leads to finite-size corrections that need to be estimated. On a small box, we define a trial state as a suitable modification of the Gibbs state for an ideal Bose gas. It will be convenient to use periodic boundary conditions instead of the Dirichlet boundary conditions that naturally arise when confining the particles to the small boxes, and we shall estimate the effect of this change of boundary conditions on the free energy.  In order to obtain the correct interaction energy, we consider  an ideal gas of reduced  density $\rho - \rho_{\mathrm{s}}$, and add a condensate of density $\rho_{\mathrm{s}}$ as a coherent state. Moreover, because of the short-range interaction it is necessary to add a correlation structure via a Jastrow-type  product function \cite{Jastrow1955}, involving the solution to the zero-energy scattering equation defining the scattering length, and we need to estimate its effect on the norms of the eigenfunctions of the state. We then proceed with estimating the energy and  entropy  of the trial state on the small box. A suitable choice of the various parameters leads to the stated bound \eqref{eqn-thm-ub}.

\subsection{Preliminaries}\label{ss:scatt}

In this subsection we present some tools that will be needed in our proof. We first state a lemma for approximating sums by integrals. A second lemma concerns properties of the zero-energy two-body scattering solution,  and finally we discuss the variational definition of the free energy in the canonical and grand canonical setting, as well as the equality of these two definitions in the thermodynamic limit.

The following lemma is a variation of   \cite[Lemma~4]{Seiringer06}. 

\begin{lem} \label{lem-approx-traces} 
	Let $f : \Rr_+ \to \Rr_+$ be a monotone decreasing function. With $-\Delta$ the Laplacian with periodic boundary conditions on $[0,\ell]^2$, we have
	\begin{equation}\label{eql2}
	\frac{\ell^2}{4\pi^2} \int_{\Rr^2} \left( 1 - \frac{4}{\ell |p|} \right) f(p^2) \de p \leq \Tr f(-\Delta) \leq \frac{\ell^2}{4\pi^2}\int_{\Rr^2} \left( 1 + \frac{4}{\ell |p|} \right) f(p^2) \de p + f(0).
	\end{equation}
	\begin{proof}
	The spectrum of $-\Delta$ is $\sigma(-\Delta) = [(2\pi/\ell)\Zz]^2$, hence 
	\begin{equation}
	\Tr f(-\Delta) = \sum_{p \in (2\pi/\ell) \Zz^2} f(p^2).
	\end{equation}
	Consider a decomposition of the plane into squares of side length $2\pi/\ell$. Since $f$ is monotone decreasing,  the smallest value of $f(p^2)$ for $p$ in such a square is obtained at the corner that is farthest away from the origin. Thus  the sum over the points $p$ that do not lie on a coordinate axis (i.e., the points $p = (p_1,p_2)$ for which neither $p_1 \neq 0$ nor $p_2 \neq 0$) is the lower Riemann sum to the integral of $f$ over the plane:
		\begin{equation}
		\sum_{p \in (2\pi/\ell)\Zz^2} f(p^2) \leq \frac{\ell^2}{4\pi^2} \int_{\Rr^2} f(p^2) \de p + \sum_{p \in \text{axes}} f(p^2).
		\end{equation}
		Similarly, we can estimate the sum over the axes by a one-dimensional integral as
		\begin{equation}
		\sum_{p \in \text{axes}} f(p^2) = f(0) + 4 \sum_{ n =1}^\infty  f( (2\pi n/\ell)^2) \leq f(0) + \frac{2 \ell}{\pi} \int_0^\infty f(p^2) \de p = f(0) + \frac{\ell}{\pi^2} \int_{\Rr^2} \frac{f(p^2)}{|p|} \de p.
		\end{equation}
		In combination, this yields the second inequality in \eqref{eql2}.
		
		For the lower bound we proceed in a similar fashion. We use that $f(p^2)$ attains its largest value at the corners that lie closest to the origin, and conclude that the sum over all points is the upper Riemann sum to the integral over the plane without the region
		\begin{equation}
		\mathcal{G} = \left\{ (p_1, p_2) \in \Rr^2 : 0 < p_1 < \frac{2\pi}{\ell} \text{ or } - \frac{2\pi}{\ell} < p_2 < 0 \right\}. 
		\end{equation}
		See Fig.~\ref{fig-trace-approximation} for an illustration.
		In particular, 
		\begin{equation}
		\int_{\Rr^2} f(p^2) \de p  
		\leq  \frac{4\pi^2}{\ell^2} \sum_{p \in (2\pi/\ell)\Zz^2} f(p^2) + \int_{\mathcal{G}} f(p^2) \de p.
		\end{equation}
		%
%
%
		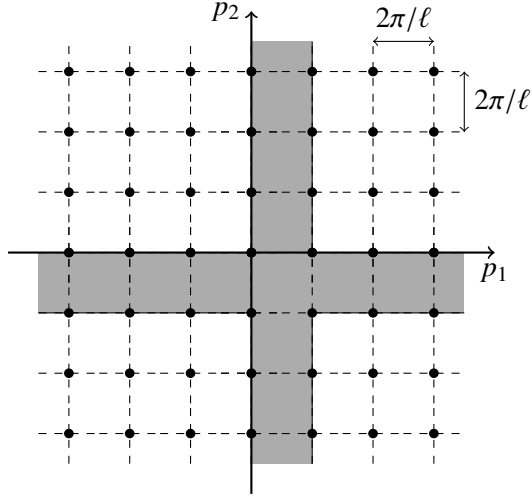
\begin{figure}[htb]
			\centering
			\begin{tikzpicture}[scale=0.8]
			\fill[fill=gray!65!white] (0,-3.5) rectangle(1,3.5);
			\fill[fill=gray!65!white] (-3.5,-1) rectangle(3.5,0);		
			\draw[<-,thick] (0,4) node[left] {$p_2$} -- (0,-4);
			\draw[->,thick] (-4,0) -- (4,0) node[below] {$p_1$};
			\foreach \x in {0,...,3} 
			{
				\foreach \y in {0,...,3}
				{
					\draw[fill=black] (\x,\y) circle(2pt);
					\draw[fill=black] (-\x,\y) circle(2pt);
					\draw[fill=black] (\x,-\y) circle(2pt);
					\draw[fill=black] (-\x,-\y) circle(2pt);
					\draw[dashed] (\x,\y) ++ (-0.5,0) -- ++ (1,0);
					\draw[dashed] (-\x,\y) ++ (-0.5,0) -- ++ (1,0);
					\draw[dashed] (\x,-\y) ++ (-0.5,0) -- ++ (1,0);
					\draw[dashed] (-\x,-\y) ++ (-0.5,0) -- ++ (1,0);
					\draw[dashed] (\x,\y) ++ (0,-0.5) -- ++ (0,1);
					\draw[dashed] (-\x,\y) ++ (0,-0.5) -- ++ (0,1);
					\draw[dashed] (\x,-\y) ++ (0,-0.5) -- ++ (0,1);
					\draw[dashed] (-\x,-\y) ++ (0,-0.5) -- ++ (0,1);
				}
			}
			\draw (-3.5,-1) -- (0,-1);
			\draw (3.5,-1) -- (1,-1) -- (1,-3.5);
			\draw (1,3.5) -- (1,0);
			\draw [<->] (3.5,2) -- (3.5,3) node[right,pos=0.5] {$2\pi/\ell$};
			\draw [<->] (2,3.5) -- (3,3.5) node[above,pos=0.5] {$2\pi/\ell$};
			\end{tikzpicture}
			\caption{Illustration of the method of proof of Lemma \ref{lem-approx-traces}. For the upper bound we use the points that do not lie on the coordinate axes, while for the lower bound we estimate the sum over all points by the integral over the plane without the gray region $\mathcal{G}$.}
			\label{fig-trace-approximation}
		\end{figure}
%
		We estimate the integral over $\mathcal{G}$ by four times the integral over the strip $\{ 0 < p_1 < 2\pi/\ell, \, p_2 > 0 \}$:
		\begin{align*}
		\int_{\mathcal{G}} f(p^2) \de p &\leq 4 \int_0^{2\pi/\ell} \int_0^\infty f(p_1^2 + p_2^2) \de p_2 \de p_1\\
		&\leq \frac{8\pi}{\ell} \int_0^\infty f(p_2^2) \de p_2 = \frac{4}{\ell} \int_{\Rr^2} \frac{f(p^2)}{|p|} \de p. \asterisknum
		\end{align*}
		Combining the  previous two estimates, we obtain the first inequality in \eqref{eql2}. 
	\end{proof}
\end{lem}

Recall that the scattering length of an interaction potential $v$ can be defined by minimizing the functional
\begin{equation} \label{eqn-0-energy-scattering}
 \int_{|x|<R} \left( 2 | \nabla g(x)|^2 + v(|x|) |g(x)|^2 \right) \de x  
\end{equation}
over functions satisfying the boundary condition $g(x) = 1$ on the sphere $|x|=R$ (see \cite[Appendix~A]{LY2001} for  details). The minimal value equals $4\pi / \ln (R/a) $, and this defines the scattering length $a$ in case $v$  is supported on a ball of radius $R_0 < R$. The unique minimizer $g_0$ of \eqref{eqn-0-energy-scattering} then satisfies $g_0(r) = \ln(r/a)/\ln(R/a)$ for $R_0 < r < R$. If  $v$ has infinite range, this definition yields the scattering length of $v(|x|) \theta(R-|x|)$, denoted by $a_R$. The latter is increasing in $R$, and  the scattering length $a$ of $v$ is  obtained by taking $R\to \infty$.  As already mentioned above, the finiteness of  $a$  is  equivalent to integrability of $v(|x|)(\ln|x|)^2$ outside a ball  (see \cite[Lemma~1]{Seiringer12}). 

The minimizer $g_0$ of \eqref{eqn-0-energy-scattering} has the following properties. 

\begin{lem} \label{lem-g'}
	
	Let $g_0$ be the minimizer of \eqref{eqn-0-energy-scattering} subject to the boundary condition $g_0(R) = 1$. Then the following holds:
	
	\begin{enumerate}
		\item For all $0 < r \leq R$
		\begin{equation}\label{gg0}
		g_0(r) \geq \left[ \frac{\ln(r/a)}{\ln(R/a)} \right]_+
		\end{equation}
		
		\item  $g_0$ is a monotone nondecreasing function of $r$.
		
		\item The integral of the derivative of $g_0$ satisfies the bound
		\begin{equation}
		\int_{|x|<R} g'_0(|x|)\de x \leq \frac{2 \pi R}{\ln(R/a)}.
		\end{equation}
	\end{enumerate}
\end{lem}	
	\begin{proof}
		For the proof of the first two properties see \cite[proof of Lemma~A.1]{LY2001}. For the third one note that since $g_0$ is a radial function we can integrate by parts in the radial variable, and then use \eqref{gg0}:
		\begin{equation}
		\int_0^R r g_0'(r) \de r = R g_0(R) - \int_0^R g_0(r) \de r \leq R - \int_{a}^R \frac{\ln(r/a)}{\ln(R/a)} \de r = \frac{R-a}{\ln (R/a)} \leq \frac{R}{\ln(R/a)}. 
		\end{equation}
		Since the angular integration only gives a factor of $2\pi$, we arrive at the result.
	\end{proof}

The last tool we require is a variational formulation of the free energy, which is very useful for the purpose of proving an upper bound. We first define the free energy functional  in the canonical setting, then in the grand canonical setting and finally show  that in the thermodynamic limit the corresponding free energies agree. The canonical free energy in finite volume is defined by
\begin{equation}\label{gibbs}
F_{\text{c}}(\beta,N,L) = \inf_\Gamma \left\{ \Tr_{\mathcal{H}_N} H_N \Gamma - \beta^{-1} S(\Gamma) \right\},
\end{equation}
where $H_N$ is given in \eqref{eqn-def-N-body-Hamiltonian} and the infimum is taken over density matrices $\Gamma$ for  $N$ particles, i.e, positive trace class operators on $\mathcal{H}_N$ (defined in \eqref{def:HN}) with $\Tr_{\mathcal {H}_N} \Gamma = 1$. Here, $S(\Gamma)$ is the von Neumann entropy defined by
\begin{equation}
S(\Gamma) = - \Tr_{\mathcal{H}_N} \Gamma \ln \Gamma.
\end{equation}
The Gibbs variational principle states that  the infimum in \eqref{gibbs} is attained for  the  Gibbs state $\Gamma = \e^{-\beta H_N}/\Tr_{\mathcal{H}_N} \e^{-\beta H_N}$, hence $F_c(\beta,N,L) = -\beta^{-1} \ln \Tr_{\mathcal{H}_N} \e^{-\beta H_N}$.

The grand canonical free energy, on the other hand, is defined by
\begin{equation} \label{eqn-def-gc-free-energy}
F_{\text{gc}}(\beta,N,L) = \inf_{\Gamma} \left\{ \Tr_{\mathcal{F}} \mathbb{H} \Gamma - \beta^{-1} S(\Gamma) \right\},
\end{equation}
where the infimum is taken over density matrices $\Gamma$ on the bosonic Fock space $\mathcal{F} = \bigoplus_{M=0}^\infty \mathcal{H}_M$ with  expected number of particles equal to $N$, and  $\mathbb{H} = \bigoplus_{M=0}^\infty H_M$ with $H_0 = 0$, $H_1 = - \Delta$ and $H_N$  defined in \eqref{eqn-def-N-body-Hamiltonian} for $N \geq 2$. 
The trace in the definition of the entropy in \eqref{eqn-def-gc-free-energy} is also over $\mathcal{F}$, but we suppress this in the notation for simplicity. 
In the thermodynamic limit we obtain the free energy per unit volume as a function of the inverse temperature $\beta$ and the density $\rho$
\begin{equation}
f(\beta, \rho) = \lim_{\substack{N,L \to \infty \\ N/L^2 = \rho}} \frac{F_{\text{c}}(\beta,N,L)}{L^2}, \qquad 	f_{\text{gc}}(\beta, \rho) = \lim_{\substack{N,L \to \infty \\ N/L^2 = \rho}} \frac{F_{\text{gc}}(\beta,N,L)}{L^2}.
\end{equation}
The following is a simple consequence of the well-known equivalence of ensembles.

\begin{lem}\label{lem:cgc} 
For any $\beta>0$ and $\rho>0$
\begin{equation}
	f(\beta,\rho) = f_{\mathrm{gc}}(\beta,\rho).
\end{equation}
\end{lem}

\begin{proof}
		One trivially has $f_\text{gc}(\beta,\rho) \leq f(\beta,\rho)$ since the former is obtained by taking the infimum over a larger set. Let $\mathcal{F}^{\beta,L}(\Gamma)$ denote the grand canonical free energy functional (i.e., the right-hand side of \eqref{eqn-def-gc-free-energy} without the infimum) and  introduce the grand canonical pressure functional in finite volume
		\begin{equation}
		- L^2 \mathcal{P}_L^{\beta,\mu}(\Gamma) = \Tr_{\mathcal{F}} (\mathbb{H}- \mu \Nn) \Gamma - \beta^{-1} S(\Gamma)
		\end{equation}
		for $\mu \in \mathbb{R}$, where $\Nn$ is the particle number operator on Fock space. Maximizing this functional over all density matrices $\Gamma$, we obtain the grand canonical pressure in finite volume
		\begin{equation}
		P_L(\beta,\mu) = \sup_\Gamma \mathcal{P}_L^{\beta,\mu}(\Gamma).
		\end{equation}
		Finally, the thermodynamic pressure is defined by
		\begin{equation}
		p(\beta,\mu) = \lim_{L \to \infty} P_L(\beta,\mu).
		\end{equation}
		For any $\mu \in \Rr$ we have
		\begin{align*}
		f_\text{gc}(\beta,\rho) &= \lim_{L\to \infty} L^{-2} \inf_{\Gamma, \braket{\Nn}_\Gamma = \rho L^2} \mathcal{F}^{\beta,L}(\Gamma)\\
		&= \lim_{L \to \infty} L^{-2} \inf_{\Gamma, \braket{\Nn}_\Gamma = \rho L^2} \left( \Tr_{\mathcal{F}} (\mathbb{H}- \mu \Nn) \Gamma - \beta^{-1} S(\Gamma) + \mu \rho L^2 \right)\\
		&\geq \lim_{L \to \infty} L^{-2} \inf_{\Gamma} \left( - L^2 \mathcal{P}_L^{\beta,\mu}(\Gamma) + \mu \rho L^2 \right) 
		= \mu \rho - p(\beta,\mu), \asterisknum
		\end{align*}
		where we relaxed the condition on the expectation of the particle number operator in order to obtain a lower bound in terms of the pressure. It is  well-known  (see, e.g.,  \cite[Thm.~3.5.8]{Ruelle69}) that the canonical free energy is the Legendre transform of the pressure, and thus 
		\begin{equation}
		f_\text{gc}(\beta,\rho) \geq \sup_\mu \left( \mu \rho - p(\beta,\mu) \right) = f(\beta,\rho).
		\end{equation}
		Consequently $f(\beta,\rho) = f_\text{gc}(\beta,\rho)$.
	\end{proof}

\subsection{Changing boundary conditions}\label{ss:bc}

In this subsection we shall relate Hamiltonians with different boundary conditions. Our method is inspired by the arguments in \cite{Robinson71}. Let  $\Lambda_L=[-L/2,L/2]^2$ denote the square of side length $L$ centered at the origin.   For $0< b < L/2$, we introduce a cutoff function $h:\mathbb{R}\to [0,1]$ with the following properties. 
\begin{enumerate}
	\item $h$ is real-valued, even and continuously differentiable
	\item $h(x) = 0$ for $|x| > L/2 + b$
	\item $h(x) = 1$ for $|x| < L/2 - b$
	\item $h(x)^2 + h(x-L)^2 + h(x + L)^2 = 1$ for $-L/2 \leq x \leq L/2$
	\item $|h'(x)|^2 \leq 1/b^2$ 
	for all $x \in \Rr$
\end{enumerate}
Condition 4 can be reformulated as  antisymmetry of $y \mapsto 1/2 - |h(y - L/2)|^2$ on $[-b,b]$. For points in the plane, we shall slightly abuse notation and write $h(x) = h(x^{(1)}) h(x^{(2)})$ for  $x = (x^{(1)},x^{(2)}) \in \Rr^2$. Finally, define $V : \mathcal{H}_N(\Lambda_L) \to \mathcal{H}_N(\Lambda_{L+2b})$ by
\begin{equation} \label{eqn-mapping-V}
 (V\psi)(x_1, \ldots{}, x_N) = \psi_\text{per}(x_1, \ldots{}, x_N) \prod_{i=1}^N h(x_i),
\end{equation}
where $\psi_\text{per}$ denotes the periodic extension of $\psi$ to $\Lambda_{3L}$, 
\begin{equation}
\psi_\text{per}(x_1,\dots,x_N) =  \sum_{ \{ n_1,\dots,n_N\} \in \{-1,0,1\}^{2N}} \psi(x_1 + n_1 L, \dots, x_N+ n_N L) \prod_{k=1}^N \chi_{\Lambda_L}(x_k+n_k L).
\end{equation}
Here, $\chi_{\Lambda_L}$ denotes the characteristic function of the set $\Lambda_L$. As in \cite[Lemma~2.1.12]{Robinson71} one easily checks that $V$ is an isometry:
\begin{align}\nonumber
\| V \psi\|^2 &= \sum_{ \{ n_1,\dots,n_N\} \in \{-1,0,1\}^{2N}} \int_{\Lambda_{3L}^N} | \psi(x_1 + n_1 L, \dots, x_N+ n_N L)|^2 \prod_{k=1}^N \chi_{\Lambda_L}(x_k+n_k L) h(x_k)^2 \de x_k \\  & = \sum_{ \{ n_1,\dots,n_N\} \in \{-1,0,1\}^{2N}} \int_{\Lambda_{L}^N} | \psi(x_1, \dots, x_N)|^2 \prod_{k=1}^N h(x_k - n_k L)^2 \de x_k = \|\psi\|^2
\end{align}
where we have used that $\sum_{n\in\{-1,0,1\}^2} h(x+nL)^2 = 1$ for $x \in  \Lambda_L$. 

\begin{lem} \label{lem-D-N-bc}
	With $v$ as above, let $H^{\mathrm{D}}_{N,\Lambda_{L+2b}}$ denote the $N$-particle Hamiltonian 
	\begin{equation}\label{def:HD}
	H_{N,\Lambda_{L+2b}}^{\mathrm{D}} = -\sum_{i=1}^N \Delta_{i,L+2b}^{\mathrm{D}} + \sum_{i<j}^N v(|x_i-x_j|)
	\end{equation}
	on  $\mathcal{H}_N(\Lambda_{L+2b})$  with Dirichlet boundary condition, and let
	\begin{equation}\label{def:Hp}
	H_{N,\Lambda_{L}}^{\mathrm{per}} = -\sum_{i=1}^N \Delta_{i,L}^{\mathrm{per}} + \sum_{i<j}^N v_{\mathrm {per}}(x_i-x_j)
	\end{equation}
	be the $N$-particle Hamiltonian on $\mathcal{H}_N(\Lambda_{L})$ with periodic boundary conditions and interaction
	\begin{equation}
	v_{\mathrm{per}}(x) = \sum_{n\in \mathbb{Z}^2} v(|x+ n L|).
	\end{equation}
	For any $\psi$ in the form domain of $H_{N,\Lambda_L}^{\mathrm{per}}$ we have
	\begin{equation}
	\braket{ V\psi, H_{N,\Lambda_{L+2b}}^{\mathrm{D}} V\psi} \leq \braket{\psi, H_{N,\Lambda_L}^{\mathrm{per}} \psi} + \frac{4N}{b^2} \| \psi \|^2.
	\end{equation}
\end{lem}

\begin{proof}
For the kinetic energy, the proof is the same as the one of \cite[Lemma~2.1.12]{Robinson71}, where the case of Neumann (instead of periodic) boundary conditions is considered, and we will not repeat it here.  For the interaction energy, we note that
\begin{align}\nonumber
& \braket{ V\psi, v(|x_1-x_2|) V\psi}  \\ \nonumber & = \sum_{n_1,n_2\in \{-1,0,1\}^2} \int_{\Lambda_L^N} |\psi(x_1,\dots,x_N)|^2  v(| x_1 + n_1 L - x_2 - n_2 L|) h(x_1 + n_1 L)^2 h(x_2 + n_2 L)^2\prod_{k=1}^N \de x_k
\\  & \leq \int_{\Lambda_L^N} |\psi(x_1,\dots,x_N)|^2  v_{\mathrm{per}}(x_1  - x_2) \prod_{k=1}^N \de x_k = \braket{ \psi, v_{\mathrm{per}}(x_1-x_2) \psi} 
\end{align}
where we have used that $v\geq 0$ and condition $4$ on $h$ above. This completes the proof.
\end{proof}

\subsection{Box method}

We will now construct a trial state that we insert into the grand canonical free energy functional in \eqref{eqn-def-gc-free-energy}. As shown in Lemma~\ref{lem:cgc}, the canonical and grand canonical free energies coincide in the thermodynamic limit, hence it is legitimate to work with the latter. For some $R_0>a>0$, consider a partition of the square of size $L$ into $(L/\ell)^2$ smaller boxes of size $\ell - R_0$, separated a distance $R_0$. We consider a trial state  that is a tensor product\footnote{Strictly speaking, one should take a symmetric tensor product here; the symmetrizing has no effect, however, as all Hamiltonians considered are  local, in the sense that matrix elements vanish for functions with disjoint support.} of translates of a given state $\Gamma$ that is supported on a small box and has an average particle number $n = N (\ell/L)^2 = \rho \ell^2$. 
Note that the choice of $\ell$ is restricted by the condition that $L/\ell$ is an integer, which will no play role, however, as $L\to \infty$ while $\ell$ stays finite in the thermodynamic limit. 

With $\rho_\Gamma$ denoting the one-particle density of $\Gamma$,  the variational principle \eqref{eqn-def-gc-free-energy} implies
\begin{equation} \label{eqn-var-princ-1}
\ell ^2 f(\beta,\rho)  \leq   \Tr_{\mathcal{F}} \left( \mathbb{H}_{\Lambda_{\ell-R_0}}^{\mathrm{D}}\Gamma \right)  -  \frac 1 \beta S(\Gamma)  + \frac 12 \sum_{n\in \mathbb{Z}^2, n\neq 0}  \int v(|x-y|) \rho_\Gamma(x)\rho_\Gamma(y+n \ell) \de x  \de y,
\end{equation}
where $ \mathbb{H}_{\Lambda_{\ell-R_0}}^{\mathrm{D}}  = \bigoplus_N  {H}_{N,\Lambda_{\ell-R_0}}^{\mathrm{D}}$ with the Dirichlet Hamiltonians $ {H}_{N,\Lambda_{\ell-R_0}}^{\mathrm{D}}$ defined in \eqref{def:HD}. The last term in 
\eqref{eqn-var-princ-1} results from the interaction of particles in different boxes and vanishes if the range of $v$ is smaller than $R_0$.

We choose $\Gamma$ of the form  $\Gamma = V^* \Gamma_{\mathrm{P}}V$,  where  $V$ is  defined in \eqref{eqn-mapping-V}, and  $\Gamma_{\mathrm{P}}$ is a translation invariant state on the torus of side length $\ell-R_0-2b$, i.e., its eigenfunctions satisfy periodic boundary conditions and are in the domain of the periodic Hamiltonians $H_{N,\Lambda_{\ell-R_0 -2b}}^{\mathrm{per}}$ in   \eqref{def:Hp} for suitable particle numbers.
Lemma \ref{lem-D-N-bc} then implies
\begin{equation}
 \Tr_{\mathcal{F}} \left( \mathbb{H}_{\Lambda_{\ell-R_0}}^{\mathrm{D}}\Gamma \right)   \leq  \Tr_{\mathcal{F}}  \left( \mathbb{H}_{\Lambda_{\ell-R_0-2b}}^{\mathrm{per}} \Gamma_{\mathrm{P}}\right) + \frac{4}{b^2} \rho \ell^2
\end{equation}
with $ \mathbb{H}_{\Lambda_{\ell-R_0-2b}}^{\mathrm{per}}  = \bigoplus_N  {H}_{N,\Lambda_{\ell-R_0-2b}}^{\mathrm{per}}$. 
We  further use that the von Neumann entropy is invariant under isometries, hence $S(\Gamma) = S(\Gamma_{\mathrm{P}})$. The state $\Gamma_{\mathrm{P}}$ has a constant particle density $n (\ell-R_0-2b)^{-2} = \rho( 1- R_0/\ell -2b/\ell)^{-2}$, and it is not difficult to see that the density $\rho_\Gamma$ equals that number times the function $h$ used in the construction of $V$ in Sec.~\ref{ss:bc}. In particular, $\rho_\Gamma\leq    \rho( 1- R_0/\ell -2b/\ell)^{-2}$, and since the boxes are separated a distance $R_0$, we have
\begin{equation}
\sum_{n\in \mathbb{Z}^2, n\neq 0}  \int v(|x-y|) \rho_\Gamma(x)\rho_\Gamma(y+n \ell) \de x  \de y \leq \frac {\rho^2 (\ell - R_0)^2}{ ( 1- R_0/\ell -2b/\ell)^{4}} \int_{|x|>R_0} v(|x|) \de x .
\end{equation}
We shall bound the right side  as 
\begin{equation}\label{244}
\int_{|x|>R_0} v(|x|) \de x \leq [\ln (R_0/a)]^{-2} \int_{|x|>R_0} v(|x|) [\ln (|x|/a)]^2 \de x 
\end{equation}
for $R_0>a$, and recall that the last integral is finite for interaction potentials $v$ with finite scattering length (and hence goes to zero as $R_0$ becomes large).

We  conclude that
\begin{align} \nonumber 
\ell ^2 f(\beta,\rho)  & \leq   \Tr_{\mathcal{F}}  \left( \mathbb{H}_{\Lambda_{\ell-R_0-2b}}^{\mathrm{per}} \Gamma_{\mathrm{P}}\right)    -  \frac 1 \beta S(\Gamma_{\mathrm{P}})  + \frac{4}{b^2} \rho \ell^2   \\ & \quad  + \frac 12 \frac {\rho^2 (\ell - R_0)^2}{ ( 1- R_0/\ell -2b/\ell)^{4}} \frac {  1 } {  [\ln (R_0/a)]^{2}}  \int_{|x|>R_0} v(|x|) [\ln (|x|/a)]^2 \de x .
\label{eqn-upper-bound-periodic-b}
\end{align}
We are left with the task of finding an upper bound on the free energy of a finite system of size $\ell - R_0 - 2b$ with periodic boundary conditions, containing an average particle number $\rho \ell^2$. This will be done in the next section. The trial state that we will use is constructed from a Gibbs state of a non-interacting gas, a manually added condensate and a product function introducing the appropriate correlations due to the particle interactions.

\subsection{Periodic trial state on a finite box}

Denote $\tilde \ell = \ell - R_0 -2b$ for simplicity, and consider the Gibbs state of an ideal Bose gas
\begin{equation}\label{def:la}
\Gamma_{\mathrm{G}} = \sum_\alpha \lambda_\alpha |\psi_\alpha \rangle \langle \psi_\alpha|, \quad \lambda_\alpha = \frac{\e^{-\beta (E_\alpha - \mu N_\alpha)}}{\sum_{\alpha'} \e^{-\beta (E_{\alpha'} - \mu N_{\alpha'})}}
\end{equation}
where $E_\alpha$ and $N_\alpha$ are the  eigenvalues of $\mathbb{H}_0 := \bigoplus_N \sum_{i=1}^N (-\Delta^{\mathrm{per}}_{i,\tilde \ell})$ and the number operator $\Nn =\bigoplus_N N$ for an eigenstate $\psi_\alpha$, and $\mu<0$ is chosen such that $n_{\mathrm{G}} := \sum_\alpha \lambda_\alpha N_\alpha \leq n = \rho\ell^2$. We will in fact take 
\begin{equation}
n_{\mathrm{G}} = n  \min \left\{ 1, \frac {\beta_\mathrm{c}}{\beta} \right\} = \ell^2 ( \rho -\rho_{\mathrm{s}}) 
\end{equation}
with $\beta_\mathrm{c}$ defined in \eqref{eqn-def-beta-c}. 

For reasons that will become apparent below, we introduce a cutoff on the number of particles in $\Gamma_{\mathrm{G}}$ by restricting the sum in \eqref{def:la} to the set
\begin{equation}
\mathcal{A} = \{ \alpha : N_\alpha < \mathcal{N} \}
\end{equation}
for some parameter $\mathcal{N}>0$  to be chosen later. In order for the  state to still have trace one, we need to modify the coefficients $\lambda_\alpha$ and use instead
\begin{equation}	\label{def:lat}
\tilde \lambda_\alpha = \frac{\lambda_\alpha}{\sum_{\alpha' \in \mathcal{A}} \lambda_{\alpha'}}
\end{equation}
satisfying $\sum_{\alpha \in \mathcal{A}} \tilde \lambda_\alpha = 1$. 

We use the notation $a_p$ and $a_p^\dagger$ for the annihilation and creation operators of a plane wave of momentum $p$ on Fock space. 
For $z \in \Cc$, $D_z$ denotes the coherent state (Weyl) operator for the $p=0$ mode
\begin{equation}
D_z = \exp \left(z a_0^\dagger - \overline{z} a_0 \right).
\end{equation}
It acts as a shift operator on the $p=0$ mode creation/annihilation operators and as identity on the other modes, 
\begin{equation}\label{Ws}
D_z^\dagger a_p D_z = a_p + z \delta_{p,0}.
\end{equation}

The trial state we shall use is 
\begin{equation} \label{eqn-def-trial-state}
\Gamma_{\mathrm{P}} = \sum_{\alpha \in \mathcal{A}} \tilde \lambda_\alpha \frac{ |f D_z \psi_\alpha \rangle \langle f D_z \psi_\alpha|}{\| f D_z \psi_\alpha \|^2}. 
\end{equation}
Here, $f$ is an operator on Fock space that acts in the sector of particle number $k \geq 2$ as
\begin{equation}
f_k =  f P_k = \prod_{i<j}^k g(d(x_i,x_j)) P_k,
\end{equation}
where $P_k$ is the projection onto particle number $k$, $g(r) = g_0(r)$ for $r \leq R$, $g(r) = 1$ for $r > R$ and $g_0$ is the minimizer  of  \eqref{eqn-0-energy-scattering}  with boundary condition $g_0(R) = 1$. For $k\in\{0,1\}$ we define $f_k$ to be the identity operator. The parameter $R$ will be chosen large compared to $a$ but small compared to the mean particle spacing, i.e., 
$ a \ll R \ll \rho^{-1/2}$. 

In order for $\Gamma_{\mathrm{P}}$ to have the required average particle number, we need to choose $\mu<0$ and $z\in \mathbb{C}$ such that
\begin{equation}
\Tr_{\mathcal{F}} \Nn \Gamma_{\mathrm{P}} =  \sum_{\alpha \in \mathcal{A}} \tilde \lambda_\alpha N_\alpha + |z|^2 \stackrel{!}{=} n. 
\end{equation}
The total particle number is given as the sum of particles in the (modified) thermal Gibbs state, $\tilde n_{\mathrm{G}} := \sum_{\alpha \in \mathcal{A}} \tilde \lambda_\alpha N_\alpha$, and $|z|^2$ particles in the added condensate. Since
\begin{equation}
\sum_{\alpha,\alpha'} ( N_\alpha - N_{\alpha'} ) ( \chi_{N_\alpha < \mathcal{N} } - \chi_{N_{\alpha'} < \mathcal{N} }) \lambda_\alpha \lambda_{\alpha'} \leq 0 
\end{equation}
we have $\tilde n_{\mathrm{G}} \leq n_{\mathrm{G}}$, hence $|z|^2 \geq n - n_{\mathrm{G}} \geq 0$.

We divide the calculation of the upper bound on the free energy of the trial state $\Gamma_{\mathrm{P}}$ into four lemmas. We start with an estimate on the norms appearing in the denominator in \eqref{eqn-def-trial-state}.

\begin{lem} \label{lem-norm-estimate}
	For all $\alpha \in \mathcal{A}$, we have the lower bound
	\begin{equation}\label{def:B1}
	\| f D_z \psi_\alpha \|^2 \geq 1 - \frac{\pi R^2}{2 \tilde \ell^2} \left( |z|^4 + 4 |z|^2 \mathcal{N} + 2 \mathcal{N}^2\right) =: \frac 1{B_1}.
	\end{equation}
\end{lem}
	
\begin{proof}
		We write $g(t)^2 = 1 - \eta(t)$ with the function $\eta$ having  support in $[0,R]$ and taking values between zero and one. Thus we have 
		\begin{align*}
		\| f D_z \psi_\alpha \|^2 
		&= \sum_m \int | (D_z \psi_\alpha)_m |^2 \prod_{i<j}^m ( 1 - \eta(d(x_i,x_j))) \de X_m, \asterisknum
		\end{align*}
		where $(D_z\psi_\alpha)_m$ denotes the $m$-particle component of the Fock space vector $D_z \psi_\alpha$, and $\de X_m$ is short for $\prod_{k=1}^m \de x_k$. 
		Since the $\psi_\alpha$ are normalized and $D_z$ is unitary, we can bound
		\begin{equation}
		\| f D_z \psi_\alpha \|^2 \geq 1 - \sum_m \sum_{i<j}^m \int |(D_z \psi_\alpha)_m|^2 \eta(d(x_i , x_j)) \de X_m.
		\end{equation}
		With $\rho_{\alpha,z}^{(2)}$ denoting the two-particle density of $D_z \psi_\alpha$, we have
		\begin{equation}
		\sum_m \sum_{i<j}^m \int |(D_z \psi_\alpha)_m|^2 \eta(d(x_i , x_j)) \de X_m = \frac 1 2 \int \eta(d(x, y)) \rho_{\alpha,z}^{(2)}(x,y) \de x \de y.
		\end{equation}
		In terms of the two-particle density $\rho^{(2)}_\alpha$ of $\psi_\alpha$, its one-particle density matrix $\gamma_\alpha$  and the corresponding density $\rho_\alpha(x) = \gamma_\alpha(x,x) = N_\alpha \tilde\ell^{-2}$, it can be expressed as 
		\begin{equation}
		\rho_{\alpha,z}^{(2)}(x,y)  = |z|^4 \tilde \ell^{-4} + |z|^2 \tilde \ell^{-2} \left( \gamma_\alpha(x,y) + \rho_\alpha(x) + \rho_\alpha(y) + \gamma_\alpha(y,x) \right) + \rho_\alpha^{(2)}(x,y).
		\end{equation}
		We can bound $|\gamma_\alpha(x,y)|^2  \leq  \rho_\alpha(x)\rho_\alpha(y) = N_\alpha^2 \tilde\ell^{-4}$, as well as  
		\begin{align*}
		\rho_\alpha^{(2)}(x,y) & = \frac{1}{\tilde \ell^4} \sum_{p_1,p_2,p_3,p_4} \e^{i p_1 x} \e^{i p_2 y} \e^{-i p_3 y} \e^{-i p_4 x} \braket{a^\dagger_{p_4} a^\dagger_{p_3} a_{p_2} a_{p_1}}_{\psi_\alpha}\\
		&= \frac{1}{\tilde \ell^4} \sum_p \braket{a^\dagger_p a^\dagger_p a_p a_p}_{\psi_\alpha} + \frac{1}{\tilde \ell^4} \sum_{p_1 \neq p_2} \left( 1 + \e^{i (p_1 - p_2) (x - y)} \right) \braket{a^\dagger_{p_1} a^\dagger_{p_2} a_{p_2} a_{p_1}}_{\psi_\alpha}\\
		&\leq \frac{1}{\tilde \ell^4} \sum_p n_p (n_p - 1) + \frac{2}{\tilde \ell^4} \sum_{p_1 \neq p_2}  n_{p_1} n_{p_2}\\
		&\leq \frac{2}{\tilde \ell^4} \sum_{p_1, p_2} n_{p_1} n_{p_2} = 2 \frac{N_\alpha^2}{\tilde \ell^4}, \asterisknum
		\end{align*}
		where we denoted by $n_p$ the occupation numbers $a^\dagger_p a_p \psi_\alpha = n_p \psi_\alpha$. 
		For $\alpha \in \mathcal{A}$, we have the uniform bound $N_\alpha < \mathcal{N}$ and hence
		\begin{equation}
		\| f D_z \psi_\alpha \|^2 \geq 1 - \frac{1}{2 \tilde \ell^2} \left( |z|^4 + 4 |z|^2 \mathcal{N} + 2 \mathcal{N}^2 \right) \int_{\mathbb{R}^2} \eta(|x|) \de x \geq 1 - \frac{\pi R^2}{2 \tilde \ell^2} \left( |z|^4 + 2 |z|^2 \mathcal{N} + 2 \mathcal{N}^2 \right).
		\end{equation}
		In the last inequality we estimated $\eta \leq 1$ on the disk of radius $R$. 
\end{proof}

The second lemma concerns the weights of the Gibbs state $\Gamma_{\mathrm{G}}$ restricted to the  set $\mathcal{A}$. We introduce the function
\begin{equation}
	\tau(\lambda,k) = \frac{  \e^{-\lambda} - 1}{k \lambda} \ln \left( 1 - \frac{\e^{-k \lambda}-1}{\e^{-\lambda}-1} \right)
\end{equation}
for $\lambda<0$ and $0<k<1$. 

\begin{lem} \label{lem-number-estimate}
	The restricted sum of the eigenvalues of the Gibbs state satisfies
	\begin{equation}\label{B2b}
	\sum_{\alpha \in \mathcal{A}} \lambda_\alpha \geq 1 - \exp \left( - k \beta |\mu| \left( \mathcal{N} - \tau(\beta \mu,k) n_{\mathrm{G}} \right) \right) =: \frac 1{B_2} 
	\end{equation}
	for any $0<k<1$. 
\end{lem}	
\begin{proof}
		We have
		\begin{equation}
		\sum_{\alpha \in\mathcal{A}} \lambda_\alpha = 1 - \sum_{\alpha \notin \mathcal{A}} \lambda_\alpha = 1 - \braket{\chi_{\Nn \geq \mathcal{N}}}_{\Gamma_{\mathrm{G}}}.
		\end{equation}
		The characteristic function can be bounded by an exponential function with parameter $\varkappa > 0$ as
		\begin{equation}
		\braket{\chi_{\Nn \geq \mathcal{N}}}_{\Gamma_{\mathrm{G}}} \leq \braket{\e^{\varkappa(\Nn - \mathcal{N})}}_{\Gamma_{\mathrm{G}}}.
		\end{equation}
		The latter expectation is readily obtained as 
		\begin{equation}
		\braket{\e^{\varkappa(\Nn - \mathcal{N})}}_{\Gamma_G} = \exp \left( \beta \tilde \ell^2 \left[ P_{\tilde \ell}(\beta,\mu + \varkappa/\beta) - P_{\tilde \ell}(\beta,\mu) \right] - \varkappa \mathcal{N} \right) 
		\end{equation}
		where  $P_{\tilde \ell}(\beta,\mu)$ denotes the grand canonical pressure of the ideal Bose gas in a finite volume, and we need to choose $\varkappa$ such that $\mu+ \varkappa/\beta<0$. An explicit computation gives 
		\begin{equation}
		P_{\tilde \ell}(\beta,\mu) = -\frac{1}{\beta \tilde \ell^2} \sum_{p \in (2\pi/\tilde \ell) \Zz^2} \ln \left( 1 - \e^{-\beta(p^2 - \mu)} \right). 
		\end{equation}
		Hence we find for the difference
		\begin{equation}
		P_{\tilde \ell}(\beta,\mu+\varkappa/\beta) - P_{\tilde \ell}(\beta,\mu) = - \frac{1}{\beta \tilde \ell^2} \sum_{p \in (2\pi/\tilde \ell) \Zz^2} \ln \left( 1 - \frac{\e^\varkappa - 1}{\e^{\beta(p^2 - \mu)} - 1} \right). 
		\end{equation}
		We shall choose $\varkappa = - k \beta \mu$ for $0<k<1$. It will be convenient to estimate $- \ln(1 - x) \leq \eta x$ where $\eta$ is chosen such that equality occurs for the largest $x$ under consideration. In our case that means
		\begin{equation}
		\eta = -  \frac{\e^{- \beta \mu} - 1}{\e^{-k\beta \mu} - 1}  \ln \left( 1 - \frac{\e^{- k \beta\mu} - 1}{\e^{-\beta \mu} - 1} \right).
		\end{equation}
		The sum over $p$ can then be evaluated as the density of the Gibbs state $\Gamma_{\mathrm{G}}$, 
		\begin{equation}
		P_{\tilde \ell}(\beta,\mu(1-k)) - P_{\tilde \ell}(\beta,\mu) \leq \frac{1}{\beta \tilde \ell^2} \sum_{p \in (2\pi/\tilde \ell)\Zz^2} \eta \frac{\e^{-k \beta \mu} - 1}{\e^{\beta(p^2 - \mu)}-1}
		 = \eta \left( \e^{-k \beta \mu} - 1 \right) \frac{n_{\mathrm{G}}}{\beta \tilde \ell^2} = -k \beta \mu  \tau(\beta\mu,k) \frac{n_{\mathrm{G}}}{\beta \tilde \ell^2}. 
		\end{equation}
		Hence the bound \eqref{B2b} follows.
\end{proof}

We remark that the same proof also shows that 
\begin{equation}
\sum_{\alpha\not\in \mathcal{A}} \lambda_\alpha N_\alpha \leq  \left( \mathcal{N} + \frac 1{k\beta|\mu|}\right) \exp \left( - k \beta |\mu| \left( \mathcal{N} - \tau(\beta \mu,k) n_{\mathrm{G}} \right) \right).
\end{equation}
To see this, one simply bounds
\begin{equation}
( \Nn - \mathcal{N} )\chi_{\Nn \geq \mathcal{N}} \leq \varkappa^{-1} \e^{\varkappa(\Nn - \mathcal{N})}
\end{equation}
for any $\varkappa>0$, and then proceeds as above. This allows us to derive a lower bound on $\tilde n_{\mathrm{G}}$ in terms of $n_{\mathrm{G}}$, 
\begin{align}\nonumber
\tilde n_{\mathrm{G}} & = \sum_{\alpha \in \mathcal{A}} \tilde\lambda_\alpha N_\alpha \geq \sum_{\alpha \in \mathcal{A}} \lambda_\alpha N_\alpha = n_{\mathrm{G}} - \sum_{\alpha \not\in \mathcal{A}} \lambda_\alpha N_\alpha
\\ & \geq n_{\mathrm{G}} - \left( \mathcal{N} + \frac 1{k\beta|\mu|}\right) \exp \left( - k \beta |\mu| \left( \mathcal{N} - \tau(\beta \mu,k) n_{\mathrm{G}} \right) \right), \label{nnt}
\end{align}
which will be useful later. 

In the third lemma we shall estimate the expectation value of the Hamiltonian $\mathbb{H}_{\Lambda_{\tilde \ell}}^{\mathrm{per}}$ in our trial state $\Gamma_{\mathrm{P}}$. To simplify the notation, we shall just write $\mathbb{H}$ for $\mathbb{H}_{\Lambda_{\tilde\ell}}^{\mathrm{per}} $, and denote the restrictions of $\mathbb{H}$ to the sector of particle number $m$ by $H_m$. 
We have
\begin{equation} \label{eqn-exp-energy}
\Tr_{\mathcal{F}} \mathbb{H} \Gamma_{\mathrm{P}} = \sum_{\alpha \in \mathcal{A}} \frac{\tilde \lambda_\alpha}{\| f D_z \psi_\alpha \|^2} \sum_{m} \int f_{m} \overline{(D_z \psi_\alpha)}_{m} H_{m} f_{m} (D_z \psi_\alpha)_{m} \de X_{m}, 
\end{equation}
where we denote again $\de X_{m} = \prod_{k=1}^m \de x_k$.  We have to evaluate the integrals
\begin{align*} \label{eqn-calc-ibp-1}
\int &f_{m} \overline{(D_z \psi_\alpha)}_{m} H_{m} f_{m} (D_z \psi_\alpha)_{m} \de X_{m}\\
&= \int | \nabla f_{m} (D_z \psi_\alpha)_{m}|^2 \de X_{m} + \sum_{i<j}^{m} \int v_{\mathrm{per}}(x_i- x_j) | f_{m} (D_z \psi_\alpha)_{m}|^2 \de X_{m}\asterisknum
\end{align*}
where $\nabla$ denotes the gradient with respect to $X_m=(x_1,\dots,x_m)$. 
Using integration by parts, the first term can be rewritten as
\begin{equation} \label{eqn-calc-ibp-2}
\int | \nabla f_{m} (D_z \psi_\alpha)_{m}|^2 \de X_{m} = \int \left( | \nabla f_{m}|^2 | (D_z \psi_\alpha)_{m}|^2 - f_{m}^2 \overline{(D_z \psi_\alpha)}_{m} \nabla^2 (D_z \psi_\alpha)_{m} \right) \de X_{m}.
\end{equation}
 The first term on the right-hand side, together with the potential term in \eqref{eqn-calc-ibp-1}, will yield the leading order correction to the free energy, while the second term will give the main contribution. We define
\begin{equation}
\mathcal{E} := \sum_{\alpha \in \mathcal{A}} \frac{\tilde \lambda_\alpha}{\| f D_z \psi_\alpha \|^2} \sum_m \int f_m^2 \overline{(D_z \psi_\alpha)}_m (- \nabla^2) (D_z \psi_\alpha)_m \de X_m.
\end{equation}
The ideal gas Hamiltonian $\mathbb{H}_0$ does not distinguish between a state with or without added quasi-condensate, since the latter carries no kinetic energy. In other words,  if $E_\alpha$ is the eigenvalue of $\psi_\alpha$, then we have 
\begin{equation}
\mathbb{H}_0 D_z \psi_\alpha = E_\alpha D_z \psi_\alpha.
\end{equation}
Hence
\begin{equation} \label{eqn-free-hamiltonian-energy}
\mathcal{E} = \sum_{\alpha \in \mathcal{A}} \frac{\tilde \lambda_\alpha}{\| f D_z \psi_\alpha \|^2} \sum_m \int f_m^2 \overline{(D_z \psi_\alpha)}_m E_\alpha (D_z \psi_\alpha)_m \de X_m = 
\sum_{\alpha \in \mathcal{A}} \tilde\lambda_\alpha E_\alpha.
\end{equation}

Next we evaluate the first term on the right-hand side of \eqref{eqn-calc-ibp-2}, proceeding similarly as in \cite{DS19,DSY}.   We have
\begin{equation} \label{eqn-calc-grad-1}
\nabla_{x_k} f_{m} = \nabla_{x_k} \prod_{i<j}^{m} g(d(x_i, x_j))=  f_m \sum_{l,l\neq k} \frac{\nabla_{x_k} g(d(x_l,x_k))}{g(d(x_l,x_k))}  .
\end{equation}
Hence the square of the gradient of $f_{m}$ is given by
\begin{equation} \label{eqn-gradient-split}
| \nabla f_m |^2 =  2 f_m^2 \sum_{l < k} \left(  \frac{g'(d(x_l,x_k))}{g(d(x_l,x_k))}  \right)^2 + f_m^2 \sum_k \sum_{\substack{l,l' \neq k\\ l \neq l'}}  \frac{\nabla_{x_k} g(d(x_l,x_k))\cdot \nabla_{x_k} g(d(x_{l'},x_k))}{g(d(x_l,x_k)) g(d(x_{l'},x_k))} .
\end{equation}
The first term contains the square of the derivative of $g$ and is needed to obtain the correct  interaction energy. The factor $2$ arises from restricting the sum to $l<k$ instead of $l\neq k$.  We further define 
\begin{align}\nonumber
\mathcal{I} &:= \sum_{\alpha \in \mathcal{A}} \frac{\tilde \lambda_\alpha}{\| f D_z \psi_\alpha \|^2} \sum_m \sum_{i<j}  \int \frac{ 2g'(d(x_i,x_j))^2 + v_{\mathrm{per}}(x_i - x_j) g(d(x_i,x_j))^2 }{g(d(x_i,x_j))^2} f_m^2   | (D_z \psi_\alpha)_m |^2 \de X_m
\\ 
\mathcal{R} &:= \sum_{\alpha \in \mathcal{A}} \frac{\tilde \lambda_\alpha}{\| f D_z \psi_\alpha \|^2} \sum_m \sum_k \sum_{\substack{l,l' \neq k\\ l \neq l'}} \int  f_m^2 \frac{\nabla_{x_k} g(d(x_l,x_k))\cdot \nabla_{x_k} g(d(x_{l'},x_k))}{g(d(x_l,x_k))g(d(x_{l'},x_k))}  | (D_z \psi_\alpha)_m |^2 \de X_{m}. \label{IR}
\end{align}

\begin{lem} \label{lem-energy-estimate}
	With the definitions above, we have
	\begin{equation}
	\Tr_{\mathcal{F}} \mathbb{H} \Gamma_{\mathrm{P}} = \mathcal{E + I + R},
	\end{equation}
	where $\mathcal{E}$ is given in \eqref{eqn-free-hamiltonian-energy}, 
	\begin{equation}\label{eqn-estimate-I}
		\mathcal{I} \leq \frac{B_1 B_2}{ \tilde \ell^{2}}  \left( ( |z|^2 + n_{\mathrm{G}})^2 - \tfrac 12 |z|^4  \right)  \left( \frac{4\pi}{\ln(R/a)} +  \int_{|x|> R} v(|x|) \de x \right) 
	\end{equation}
and
\begin{equation}\label{bound:R}
		\mathcal{R} \leq 24 \pi^2  B_1 B_2  \frac{(n_{\mathrm{G}} + |z|^2)^3}{\tilde \ell^4} \frac{R^2}{[\ln(R/a)]^2}, 
		\end{equation}
	with $B_1$ and $B_2$ defined in Lemma~\ref{lem-norm-estimate} and \ref{lem-number-estimate}, respectively.
\end{lem}

\begin{proof}		
		We introduce the function
		\begin{equation}
		\xi(x,y) = g'(d(x,y))^2 + \tfrac 1 2 v_{\mathrm{per}}(x-y) g(d(x,y))^2.
		\end{equation}
		In the integrand in the first term in \eqref{IR}, we can bound $f_m(X_m)^2 \leq g(d(x_i,x_j))^2$, and then 
		 integrate out all but two variables to find the two-particle density. Using in addition that $\tilde \lambda_\alpha \| fD_z \psi_\alpha\|^{-2} \leq \lambda_\alpha B_1 B_2$ according to Lemmas~\ref{lem-norm-estimate} and \ref{lem-number-estimate}, this leads to 
\begin{equation} \label{eqn-I-estimate-1}
		\mathcal{I} \leq  B_1 B_2 \int \xi(x,y) \rho_z^{(2)}(x,y) \de x \de y ,
\end{equation}		
where $\rho_z^{(2)}$ denotes the two-particle density of $D_z \Gamma_{\mathrm{G}} D_z^\dagger$. Here we  also used that $\xi\geq 0$ to add the missing terms in  the sum over $\alpha$ to obtain the full Gibbs state. 
Using Wick's theorem for the quasi-free state $\Gamma_{\mathrm{G}}$ and \eqref{Ws}, we calculate
		\begin{equation}
		\rho_z^{(2)}(x,y) = |z|^4 \tilde \ell^{-4} + |z|^2 \tilde \ell^{-2}( \rho(x) + \rho(y) + \gamma(x,y) + \gamma(y,x)) + |\gamma(x,y)|^2 + \rho(x) \rho(y)
		\end{equation}
with $\gamma$ and $\rho$ the one-particle density matrix and corresponding density of $\Gamma_{\mathrm{G}}$. We have $\rho(x) = n_{\mathrm{G}} \tilde\ell^{-2}$ and $|\gamma(x,y)| \leq n_{\mathrm{G}}\tilde\ell^{-2}$, hence
		\begin{equation}
		\rho_z^{(2)}(x,y) \leq \frac{1}{\tilde \ell^4} \left(  |z|^4 + 4 |z|^2  n_{\mathrm{G}} + 2 n_{\mathrm{G}}^2 \right) = \frac{2}{\tilde \ell^4} \left( (|z|^2 + n_{\mathrm{G}})^2 - \tfrac 12 |z|^4  \right) .
		\end{equation}
For any fixed $y \in \Lambda_{\tilde \ell}$, 
\begin{equation}
\int_{\Lambda_{\tilde \ell}}  \xi(x,y)  \de x\leq  \int_{\mathbb{R}^2} \left( g'(|x|)^2 + \tfrac 1 2 v(|x|) g(|x|)^2 \right) \de x = \frac{2\pi}{\ln(R/a_R)} + \frac 12 \int_{|x|> R} v(|x|) \de x ,
\end{equation}
where we have used that $g$ is an increasing function in the first step, and that the minimum of \eqref{eqn-0-energy-scattering} equals $4\pi / \ln (R/a_R)$ in the second. 
Here, $a_R$ denotes the scattering length of the potential $v(|x|) \theta(R-|x|)$, which satisfies $a_R\leq a$, yielding \eqref{eqn-estimate-I}. 
		
		We proceed similarly for the three-particle term $\mathcal{R}$. In terms of the three-particle density $\rho^{(3)}_{z}$ of  $D_z \Gamma_{\mathrm{G}} D_z^\dagger$, we obtain
		\begin{equation}
		\mathcal{R} \leq  B_1 B_2 \int g'(d(x,z)) g'(d(y,z)) \rho_z^{(3)}(x,y,z) \de x \de y \de z. 
		\end{equation}
		With the aid of Wick's theorem and \eqref{Ws}, one readily finds the crude bound 
		\begin{equation}
		\rho_z^{(3)}(x,y,z) \leq 6 \frac{(n_{\mathrm{G}} + |z|^2)^3}{\tilde \ell^6}.
		\end{equation}
		Applying in addition part 3 of Lemma \ref{lem-g'}, we obtain \eqref{bound:R}. 		
\end{proof}

Finally, we need to estimate the entropy of the trial state $\Gamma_{\mathrm{P}}$ in order to obtain a bound on the free energy. 

\begin{lem} \label{lem-entropy-estimate}
	We have
	\begin{equation}
	S(\Gamma_{\mathrm{P}}) \geq  - \sum_{\alpha \in \mathcal{A}} \tilde \lambda_\alpha \ln \tilde \lambda_\alpha   - \ln B_1,
	\end{equation}
	where $B_1$ 
	is defined in Lemma~\ref{lem-norm-estimate}. 
\end{lem}

\begin{proof}
		The proof follows  \cite[Lemma~2]{Seiringer06} and for the reader's convenience we repeat it here. The state $\Gamma_{\mathrm{P}}$ is of the form  $ \Gamma_{\mathrm{P}} = \sum_{\alpha \in \mathcal{A}} \tilde \lambda_\alpha P_\alpha$ for rank one projections  $ \{ P_\alpha \}$  that are not necessarily mutually orthogonal. By the concavity of the logarithm we have 
		\begin{align*}
		S( \Gamma_{\mathrm{P}})  +\sum_{\alpha \in \mathcal{A}} \tilde \lambda_\alpha \ln \tilde \lambda_\alpha    &= - \sum_{\alpha \in \mathcal{A}} \tilde \lambda_\alpha \Tr_{\mathcal{F}} P_\alpha \ln \left( \tilde \lambda_\alpha^{-1}  \Gamma_{\mathrm{P}} \right)\\
		&\geq - \sum_{\alpha \in \mathcal{A}} \tilde \lambda_\alpha \ln \Tr_{\mathcal{F}} P_\alpha \tilde \lambda_\alpha^{-1}  \Gamma_{\mathrm{P}}\\
		&\geq - \ln \Tr_{\mathcal{F}} \left( \sum_{\alpha \in \mathcal{A}} P_\alpha \Gamma_{\mathrm{P}} \right) \geq - \ln \left\| \sum\nolimits_{\alpha \in \mathcal{A}} P_\alpha \right \|. \asterisknum
		\end{align*}
Since $f\leq 1$ and the functions  $D_z \psi_\alpha$ are orthonormal, we infer from Lemma~\ref{lem-norm-estimate} that
\begin{equation}
\sum_{\alpha \in \mathcal{A}} P_\alpha = \sum_{\alpha \in \mathcal{A}} \frac{\ket{f D_z \psi_\alpha}\bra{f D_z \psi_\alpha}}{\| f D_z \psi_\alpha \|^2} \leq B_1.
\end{equation}
This concludes the proof.
\end{proof}

\subsection{Final upper bound}

Now that we have an estimate on every term appearing in the free energy functional, we are ready to state the upper bound on the free energy. 
Inserting the explicit form of $\tilde\lambda_\alpha$ in \eqref{def:la} and \eqref{def:lat} and recalling that $\tilde n_{\mathrm{G}} =\sum_{\alpha \in \mathcal{A}} \tilde\lambda_\alpha N_\alpha$, we have 
\begin{align}\nonumber
\sum_{\alpha \in \mathcal{A}} \tilde\lambda_\alpha E_\alpha + \beta^{-1} \sum_{\alpha \in \mathcal{A}} \tilde \lambda_\alpha \ln \tilde \lambda_\alpha & = \mu\tilde n_{\mathrm{G}}  - \beta^{-1} \ln \left( \sum_{\alpha'\in \mathcal{A}} \e^{-\beta (E_{\alpha'} - \mu N_{\alpha'})}  \right) \\
& \leq \mu\tilde n_{\mathrm{G}}  - \beta^{-1} \ln \left( \sum_{\alpha'} \e^{-\beta (E_{\alpha'} - \mu N_{\alpha'})}  \right) + \beta^{-1} \ln B_2 \label{uex}
\end{align}
where we used Lemma~\ref{lem-number-estimate} in the last step. The free energy of the ideal Bose gas can alternatively be written as
\begin{equation}\label{altern}
- \frac 1 \beta \ln \left( \sum_{\alpha'} \e^{-\beta (E_{\alpha'} - \mu N_{\alpha'})}  \right)   = \frac 1 \beta \sum_{p \in (2\pi/\tilde \ell)\Zz^2} \ln \left( 1 - \e^{-\beta (p^2 - \mu)} \right).
\end{equation}
We can use Lemma~\ref{lem-approx-traces} to bound the last sum in terms of the corresponding integral, with the result that 
\begin{align*} \label{eqn-approx-discrete-f0}
\eqref{altern} & \leq \frac 1 \beta  \left( \frac{\tilde \ell}{2\pi} \right)^2  \int_{\mathbb{R}^2} \ln \left( 1 - \e^{-\beta (p^2 - \mu)} \right)  \de p   - \frac{\tilde \ell}{\beta  \pi^2} \int_{\Rr^2} \frac{1}{|p|} \ln \left( 1 - \e^{- \beta(p^2 - \mu)} \right) \de p  \\ & \leq  \frac 1 \beta  \left( \frac{\tilde \ell}{2\pi} \right)^2  \int_{\mathbb{R}^2} \ln \left( 1 - \e^{-\beta (p^2 - \mu)} \right)  \de p + C  \frac{\tilde \ell }{\beta^{3/2}} \asterisknum
\end{align*}
for some $C>0$, where we used $\mu<0$ in the last bound. In particular,
since 
\begin{equation} \label{eqn-def-ideal-gas-free-energy}
f_0(\beta,\rho) = \sup_{\mu \leq 0} \left\{ \mu \rho + \frac{1}{4 \pi^2 \beta} \int_{\Rr^2} \ln \left( 1 - \e^{- \beta(p^2 - \mu)} \right) \de p \right\}
\end{equation}
we obtain
\begin{equation}
\eqref{uex} \leq \tilde \ell^2  f_0(\beta, \tilde n_{\mathrm{G}} \tilde \ell^{-2}) + \beta^{-1} \ln B_2 +  C  \frac{\tilde \ell }{\beta^{3/2}}.
\end{equation}

Using Lemmas~\ref{lem-energy-estimate} and~\ref{lem-entropy-estimate}, we thus have the following upper bound on the free energy in finite volume of the trial state $\Gamma_{\mathrm{P}}$:
\begin{align*} \label{eqn-ub-finite-volume1}
 \Tr_{\mathcal{F}}  \left( \mathbb{H}_{\Lambda_{\tilde\ell}}^{\mathrm{per}} \Gamma_{\mathrm{P}}\right)    -  \frac 1 \beta S(\Gamma_{\mathrm{P}})    &\leq 
 \tilde \ell^2  f_0(\beta, \tilde n_{\mathrm{G}} \tilde \ell^{-2}) + \beta^{-1} \ln B_1 B_2 +  C  \frac{\tilde \ell }{\beta^{3/2}} \\
&\quad + \frac{B_1 B_2}{ \tilde \ell^{2}}  \left( (|z|^2 + n_{\mathrm{G}})^2 - \tfrac 12 |z|^4  \right)  \left( \frac{4\pi}{\ln(R/a)} +  \int_{|x|> R} v(|x|) \de x \right)  \\
&\quad +  24 \pi^2  B_1 B_2  \frac{(n_{\mathrm{G}} + |z|^2)^3}{\tilde \ell^4} \frac{R^2}{[\ln(R/a)]^2}. \asterisknum
\end{align*}
The last term in the second line can be bounded as in \eqref{244}. 
In combination with \eqref{eqn-upper-bound-periodic-b} this gives the final upper bound 
\begin{align*} \label{eqn-ub-finite-volume}
 f(\beta,\rho)    &\leq 
 (1- R_0/\ell -2 b/\ell)^2  f_0(\beta, \tilde n_{\mathrm{G}} \tilde \ell^{-2}) + \ell^{-2} \beta^{-1} \ln B_1 B_2 +   \frac{C}{\ell\beta^{3/2}} \\
&\quad + \frac{B_1 B_2}{ \ell^2 \tilde \ell^{2}}  \left( ( |z|^2 + n_{\mathrm{G}})^2 - \tfrac 12 |z|^4  \right)  \left( \frac{4\pi}{\ln(R/a)} +\frac{1}{[\ln(R/a)]^2 }  \int_{|x|> R} v(|x|) [\ln(|x|/a)]^2 \de x \right)  \\
&\quad +  24 \pi^2  B_1 B_2  \frac{(n_{\mathrm{G}} + |z|^2)^3}{\ell^2 \tilde \ell^4} \frac{R^2}{[\ln(R/a)]^2} + \frac{4 \rho }{b^2}  \nonumber  \\ & \quad     + \frac 12 \frac { \rho^2 (1 - R_0/\ell)^2}{ ( 1- R_0/\ell -2b/\ell)^{4}} \frac { 1} {[\ln (R_0/a)]^{2}}  \int_{|x|> R_0} v(|x|) [\ln(|x|/a)]^2 \de x  . \asterisknum
\end{align*}
We shall choose the parameters such that $\tilde n_{\mathrm{G}}\tilde \ell^{-2} \geq n_{\mathrm{G}} \ell^{-2} = \rho - \rho_{\mathrm{s}}$, hence $f_0(\beta, \tilde n_{\mathrm{G}} \tilde \ell^{-2})  \leq f_0(\beta,\rho-\rho_{\mathrm{s}})$. Note that $f_0(\beta,\rho -\rho_{\mathrm{s}}) \sim \beta^{-2}$ for $\beta \rho \gtrsim 1$.
 Moreover, we can use \eqref{nnt}  to give an upper bound on $n_{\mathrm{G}}$ in terms of $\tilde n_{\mathrm{G}}=n-|z|^2$.  It remains to choose the free parameters $\ell$, $b$, $R_0$, $R$ and $\mathcal{N}$.  
In order to estimate the error stemming from $B_2$ in \eqref{B2b}, we need bounds on the chemical potential $\mu$, which will be derived in the next section.

\subsection{Effective chemical potential}

From now on, we shall use the short hand notation
\begin{equation}\label{def:sigma}
\sigma := | \ln a^2 \rho |.
\end{equation}
Recall that the chemical potential $\mu$ was chosen such that 
\begin{equation}
n_{\mathrm{G}} = \sum_{p \in (2\pi/\tilde \ell)\Zz^2}  \frac 1{\e^{\beta(p^2 - \mu)}-1} = \rho\ell^2 \min\left\{ 1, \frac{\ln \sigma}{4\pi \beta \rho} \right\} 
\end{equation}
where the last fraction is nothing but $\beta_\mathrm{c}/\beta$. The trivial lower bound $n_{\mathrm{G}} \geq 1/(\e^{-\beta \mu}-1)$ implies that 
$-\beta \mu \gtrsim 1$ if $n_{\mathrm{G}}\lesssim 1$, and $-\beta \mu \gtrsim n_{\mathrm{G}}^{-1}$ if $n_{\mathrm{G}}\gtrsim 1$.

Let us further consider  the case $n_{\mathrm{G}}\gtrsim 1$. 
A more accurate lower bound on $n_{\mathrm{G}}$ can be obtained with the aid of Lemma~\ref{lem-approx-traces}. It implies that
\begin{equation}
n_{\mathrm{G}} \geq \frac{\tilde\ell^2}{4\pi^2} \int_{\mathbb{R}^2} \frac 1{\e^{\beta(p^2 -\mu)} - 1} \left( 1 - \frac 4{\tilde \ell |p|}\right) \de p \geq - \frac{\tilde \ell^2}{4\pi \beta} \ln\left( 1 - \e^{\beta \mu}\right) - \frac {C \tilde \ell}{\beta \sqrt{|\mu|}}.
\end{equation}
Using $-\beta\mu \gtrsim n_{\mathrm{G}}^{-1}$ on the last term, we obtain
\begin{equation}
 - \ln\left( 1 - \e^{\beta \mu}\right) \leq  \frac{4\pi \beta}{\tilde \ell^2} n_{\mathrm{G}} +  \frac{C \beta^{1/2}}{\tilde \ell}   n_{\mathrm{G}}^{1/2} \leq \frac{\ell^2}{\tilde \ell^2}\ln \sigma  +  C \frac{\ell}{\tilde \ell}  \left( \ln \sigma\right)^{1/2}.
\end{equation}
We will choose the parameters such that $\ell / \tilde \ell = 1+o(1)$ as $\sigma\to \infty$, hence $-\beta \mu \geq \sigma^{-1 + o(1)}$. 

In order to control the error term in \eqref{B2b}, we also need a bound on $\tau(\beta\mu,k) n_{\mathrm{G}}$ for some fixed $0<k<1$, say $k=1/2$. For bounded $\beta \mu$, $\tau(\beta\mu,k)$ is bounded, but as $\beta\mu\to -\infty$, it diverges as $\e^{-k\beta\mu}/(k\beta|\mu|)$. On the other hand, for $\beta|\mu| \gtrsim 1$, Lemma~\ref{lem-approx-traces} readily implies that 
\begin{equation}
n_{\mathrm{G}} \leq \frac{\tilde\ell^2}{4\pi^2} \int_{\mathbb{R}^2} \frac 1{\e^{\beta(p^2 -\mu)} - 1} \left( 1 + \frac 4{\tilde \ell |p|}\right) \de p + \frac 1{\e^{-\beta\mu}-1} \lesssim \left( \tilde\ell^2 \beta^{-1} + 1 \right) \e^{\beta \mu}
\end{equation}
and hence, in particular, $\tau(\beta\mu,k) n_{\mathrm{G}}$ is bounded above by $(\tilde\ell^2 \beta^{-1} + 1)$ for $\beta|\mu|\gtrsim 1$. Since $n_{\mathrm{G}} \leq n = \rho\ell^2$, we conclude that the bound
\begin{equation}\label{bound:tau}
\tau(\beta\mu,k) n_{\mathrm{G}} \lesssim 1 + \ell^2 \rho
\end{equation}
holds uniformly in $\beta\rho \gtrsim 1$ for fixed $0<k<1$.

\subsection{Choice of parameters}

We are now ready to choose the free parameters in our upper bound. Recall the definition \eqref{def:sigma}. We shall choose $R^2\rho < 1$, hence we can write
\begin{equation}\label{Rae}
\ln(R/a) = \frac{1}{2} \left( | \ln a^2 \rho | - | \ln R^2 \rho | \right) = \frac \sigma 2 \left( 1 - \frac {|\ln R^2\rho|}{\sigma} \right). 
\end{equation}
We shall choose $R$ such that $|\ln R^2\rho| \ll \sigma$. 

Let us start with the choice of $b$. The error terms involving $b$ are of the order 
\begin{equation}
\rho b^{-2} +  b \ell^{-1} \left( \beta^{-2} + \rho^2 \sigma^{-1} \right) \lesssim   \rho b^{-2} + \rho^2 b \ell^{-1} 
\end{equation}
for $\beta\rho \gtrsim 1$, 
which leads to the choice $b^3 \sim  \ell \rho^{-1} $, and hence an error of the order $\rho^2 (\ell^2\rho)^{-1/3}$. 
We shall choose $\mathcal{N} = A \rho \ell^2$ for some large enough $A$ (of order $1$) to be determined. The main error terms involving $\ell$ are thus, in addition to $\rho^2 (\ell^2\rho)^{-1/3}$,
\begin{equation}\label{ell2}
R^2 \rho^4 \sigma^{-1} \ell^2  \quad \text{and} \quad \beta^{-3/2} \ell^{-1}.
\end{equation} 
The most relevant term turns out to be the first one, leading to the choice $\ell^2 \rho = (R^2\rho)^{-3/4} \sigma^{3/4} $ and an error of the order 
\begin{equation}\label{R1}
\frac {\rho^2}{\sigma} ( R^2 \rho \sigma^3)^{1/4}  .
\end{equation}
The other main error terms involving $R$ are
\begin{equation}
\frac {\rho^2}{\sigma} \left( \frac {|\ln R^2 \rho|}\sigma  + R^2 \rho  + \frac {R^2 \rho  \sigma}{\beta \rho} \right)
\end{equation}
of which the first is the most relevant, the others being small compared to \eqref{R1}. We equate it with \eqref{R1}, leading to the choice 
$R^2\rho \sim  \sigma^{-7}$ 
and a resulting error term
\begin{equation}\label{main:error:term}
\frac {\rho^2}{\sigma^2 } \ln  \sigma .
\end{equation}
The only parameter left to choose is $R_0$, and we can take $R_0^2 \rho \sim \sigma^2$. 

Let us summarize the choice of parameters. We have 
\begin{equation}
R^2\rho \sim \sigma^{-7} \ , \quad \ell^2 \rho \sim \sigma^{6} \ , \quad b^2\rho \sim  R_0^2 \rho \sim  \sigma^2
\end{equation}
and $\mathcal{N} = A \sigma^6$ for suitable $A$ large enough. Let us now examine the various terms in \eqref{eqn-ub-finite-volume}. 
We have $b/\ell \sim R_0/\ell\sim \sigma^{-2}$, leading to an error of at most $\beta^{-2} \sigma^{-2} \lesssim \rho^2 \sigma^{-2}$ from the prefactor multiplying $f_0$. Since $|z|^2 \leq \rho\ell^2$, we have $B_1 = 1 + O(R^2 \rho \ell^2 \rho) = 1 + O(\sigma^{-1})$ from \eqref{def:B1}, hence $\ell^{-2} \beta^{-1} \ln B_1 \lesssim \rho^2 \sigma^{-7}$. 
For $B_2$ in \eqref{B2b}, we use that $-\beta \mu \geq \sigma^{-1 + o(1)}$, as argued in the previous section, as well as \eqref{bound:tau}. This implies that for an appropriate choice of $A>0$ we have $B_2 = 1 + O(\sigma^{-\infty})$, hence all error terms involving $B_2$ are negligible compared to \eqref{main:error:term}. 
Similarly, we can give an upper bound on $|z|^2 = n - \tilde n_{\mathrm{G}} =  \ell^2 \rho_{\mathrm{s}} + n_{\mathrm{G}}  - \tilde n_{\mathrm{G}}$. From \eqref{nnt} and  $-\beta \mu \geq \sigma^{-1 + o(1)}$ we conclude in fact that $n_{\mathrm{G}}  - \tilde n_{\mathrm{G}} \lesssim  O(\sigma^{-\infty})$. Note that this also implies  that $\tilde n_{\mathrm{G}}\tilde \ell^{-2}  \geq n_{\mathrm{G}} \ell^{-2}  = \rho - \rho_{\mathrm{s}}$ for large enough $\sigma$, as claimed after \eqref{eqn-ub-finite-volume}, at least as long as $n_{\mathrm{G}} \gtrsim O(\sigma^{-K})$ for some (arbitrary) $K>0$. This condition holds if $\beta \rho \lesssim \sigma^{K}$ for some $K>0$. For larger $\beta\rho$, we simply use that $f_0$ contributes at most $\beta^{-2} \lesssim \rho^2 \sigma^{-K}$ to the free energy, and is hence negligible for $K$ large enough. 

Using also \eqref{Rae}, we conclude from \eqref{eqn-ub-finite-volume} with this choice of parameters that
\begin{equation}
f(\beta,\rho) \leq f_0(\beta,\rho - \rho_{\mathrm{s}}) + \frac{4\pi} {\sigma} \left( 2  \rho^2 - \rho_{\mathrm{s}}^2 \right) + \frac{C \rho^2}{\sigma} \left( \frac{\ln \sigma}{\sigma} + \frac 1 \sigma \int_{|x|\geq a (C a^2 \rho)^{-1/2} \sigma^{-7/2}} v(|x|) [\ln(|x|/a)]^2 \de x\right)
\end{equation}
for some universal constant $C>0$ and $\sigma$ large. 
This concludes the proof of Theorem~\ref{thm-ub}.

\vspace{0.5cm}

\textit{Acknowledgments.} We thank Andreas Deuchert for helpful discussions. Financial support by the European Research Council (ERC) under the European Union's Horizon 2020 research and innovation programme (grant agreement No.~694227) is gratefully acknowledged.




\end{document}